\documentclass[letterpaper, 10 pt, conference]{ieeeconf}  

\newtheorem{proposition}{\textbf{Proposition}}
\newtheorem{property}{\textbf{Property}}
\newtheorem{remark}{\textbf{Remark}}
\newtheorem{objective}{\textbf{Objective}}

\usepackage{amsmath}
\usepackage{amssymb} 

\renewcommand{\QED}{\QEDopen}

\usepackage{tikz}
\usetikzlibrary{arrows,matrix,positioning,patterns}
\usetikzlibrary{calc}
\usepackage{pgfplots} 
\usetikzlibrary{plotmarks}

\usepackage[utf8]{inputenc}
\usepackage{pgfplots}
\usepgfplotslibrary{groupplots}

\usepackage{url}
\usepackage{multirow}
\usepackage{colortbl}

\usepackage{algorithm} 
\usepackage{algpseudocode}

\DeclareSymbolFont{bbold}{U}{bbold}{m}{n}
\DeclareSymbolFontAlphabet{\mathbbold}{bbold}

\usepackage[labelformat=simple]{subcaption}
\DeclareCaptionLabelSeparator{periodspace}{.\quad}
\IEEEoverridecommandlockouts                              
\title{\LARGE \bf
Trajectory Planning of Automated Vehicles in Tube-like Road Segments}

\author{Mogens Graf Plessen
\thanks{MGP is with IMT School for Advanced Studies Lucca, Piazza S. Francesco 19, 55100 Lucca, Italy, {\tt\small mogens.plessen@imtlucca.it}}
}

\renewcommand{\baselinestretch}{0.96}
\begin{document}

\maketitle
\thispagestyle{empty}
\pagestyle{empty}

\begin{abstract}

This paper presents a method based on linear programming for trajectory planning of automated vehicles, combining obstacle avoidance, time scheduling for the reaching of waypoints and time-optimal traversal of tube-like road segments. System modeling is conducted entirely spatial-based. Kinematic vehicle dynamics as well as time are expressed in a road-aligned coordinate frame with path along the road centerline serving as the dependent variable. We elaborate on control rate constraints in the spatial domain. A vehicle dimension constraint heuristic is proposed to constrain vehicle dimensions inside road boundaries. It is outlined how friction constraints are accounted for. The discussion is extended to dynamic vehicle models. The benefits of the proposed method are illustrated by a comparison to a time-based method. 

\end{abstract}

\section{Introduction\label{sec_intro}}

Autonomous vehicles draw immense interest from academia, industry, media and society. Of primary concern for industrial applicability are safety issues. Secondary important issues include increased driving comfort and increased fuel efficiency by means of cooperative driving, platooning \cite{larson2015distributed}, \cite{plessen2016multi}, and car-to-infrastructure communication for anticipative driving, e.g., at traffic lights \cite{asadi2011predictive}, \cite{kamal2010board}. 

Most devised system architectures differentiate between a path planning and a path tracking layer ~\cite{urmson2008autonomous}, \cite{werling2010optimal}, \cite{ziegler2014trajectory}, \cite{gu2013focused}, \cite{qian2016optimal}, \cite{werling2012automatic},  whereby simple control laws for tracking performed well \cite{hoffmann2007autonomous}. Other tracking methods are based on model predictive control (MPC) \cite{falcone2007predictive}, \cite{brown2017safe}. See \cite{gonzalez2016review} and \cite{paden2016survey} for two recent surveys of motion planning and control techniques for automated vehicles.   

Frequently, reference velocities are assumed to be provided by a higher-level algorithm, before a road centerline is tracked \cite{di2016vehicle}, \cite{kong2015kinematic}. The motivation and contribution of this paper is twofold. The presented method enables a) simultaneous planning of velocity and steering trajectories, and b) mission planning of automated vehicles by means of capabilities for obstacle avoidance, time scheduling for the reaching of waypoints, and time-optimal traversal of road segments. The presented method is based on linear programming. It is designed to generate trajectories exploiting the complete permissible road width rather than tracking the centerline only, thereby increasing driving comfort and safety through minimized steering actuation. This is since the lower the path curvature, the higher the admissible vehicle velocity that still permits operation within tire friction limits \cite{funke2016simple}. System modeling is conducted entirely \textit{spatial-based} with space replacing time as the dependent variable  \cite{gao2012spatial}, \cite{lot2014curvilinear}, \cite{verschueren2014towards}, \cite{graf2017spatial}, \cite{karlsson2016temporal}. In addition, control rate, vehicle dimension and friction constraints are discussed. 
\begin{figure}[ht]
\centering
\input{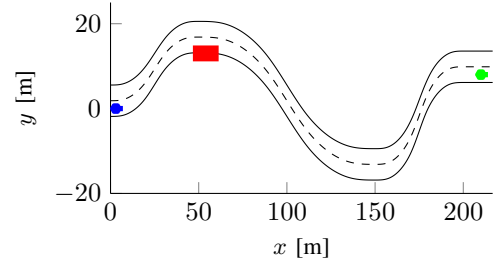}
\caption{Given the start pose (blue) of a vehicle, a path is sought avoiding any obstacles (represented by the red rectangle), traveling within corridor boundaries (black solid), respecting physical actuator constraints, and traveling within vehicle tire friction limits, such that an end pose (green) is reached. Here, two road lanes are separated by the black dashed line.}
\label{fig_ProblVisualization}
\end{figure}
The presented method is tailored for spatial-based predictive control \cite{graf2017spatial}. However, it is also illustrated how planned trajectories can serve time-based control frameworks for tracking \cite{falcone2007linear}. While the focus is on kinematic vehicle models \cite{rajamani2011vehicle} under additional consideration of friction constraints \cite{funke2016simple}, the extension to dynamic vehicle models \cite{kong2015kinematic} is also discussed.

This paper is organized as follows. Section \ref{sec_problFormulation} formulates the problem. The main contribution is given in Section \ref{sec_Main}. Section \ref{sec_numerical_experiments} states simulation results, before concluding.

\section{Problem Formulation\label{sec_problFormulation}}

Objective \ref{obj_probForm} formulates the problem addressed in this paper. See Fig. \ref{fig_ProblVisualization} for visualization.
\begin{objective}\label{obj_probForm}
Given a tube-like road described by its (real or virtual) road boundaries, develop a method for trajectory planning that permits to 
\begin{itemize}
\item find the time-optimal path trajectory,
\item time schedule the reaching of waypoints along the path,
\item simultaneously avoid any obstacles along the road, 
\end{itemize}
while being able to incorporate additional constraints.
\end{objective}
Notice that this paper does not address the combinatorial problem of deciding on which side to overtake obstacles \cite{qian2016optimal}. Instead, a driving corridor (tube-like road) is assumed, for which trajectories for traversal are sought. Trajectories here refer to state and input trajectories that can be fed to a reference trajectory tracker used in a two-layered control framework with reference trajectory planning and tracking as the two layers. Alternatively and preferably, control commands can be applied directly within a MPC-framework, combining planning and tracking in one step. State and input trajectories refer to spatial coordinates, vehicle heading (either explicitly or implicitly related to the road centerline), steering angle and traveling velocity.

\begin{figure}[ht]
\centering
\vspace{0.3cm}
\begin{tikzpicture}
\coordinate (R) at (1.8,1.55);
\coordinate (F) at ($ (R) + 4*({cos(30)},{sin(30)})$);
\coordinate (C) at ($ (R) + 1.667*({cos(30)},{sin(30)})$);
\coordinate (Ra) at ($ (C) - 2.6*({cos(30)},{sin(30)})$);
\coordinate (Raa) at ($ (C) - 2.92*({cos(30)},{sin(30)})$);
\coordinate (Rb) at ($ (C) + 3.6*({cos(30)},{sin(30)})$);
\def\psis{18} 
\def \wheelL {0.5}
\def \lineRhoSTop {2.2}
\def \lineRhoSBottom {1.7}
\draw[->] (0, 0) -- (C|-,0);
\draw[->] (0, 0) -- (0, |-C);
\node[color=black] (a) at ($ (C|-,0) + 0.25*({cos(0)},{sin(00)})$) {$x$}; 
\node[color=black] (a) at ($ (0, |-C) + 0.25*({cos(90)},{sin(90)})$) {$y$}; 
\draw[color=black] (R) -- (F);
\draw[rounded corners,black,fill=green!40,fill opacity=0.2,rotate around={30:(C)}] ($ (C) - ({2.6},{0.75})$) rectangle ($ (C) + ({3.6},{0.75})$);
\draw[color=black,|-|] ($ (Ra) - 1.1*({cos(-60)},{sin(-60)})$) -- ($ (C) - 1.1*({cos(-60)},{sin(-60)})$) node [midway, above] {$l_r$};
\draw[color=black,-|] ($ (C) - 1.1*({cos(-60)},{sin(-60)})$) -- ($ (Rb) - 1.1*({cos(-60)},{sin(-60)})$) node [midway, above] {$l_f$};
\draw[color=black,|-|] ($ (Raa) + 0.75*({cos(-60)},{sin(-60)})$) -- ($ (Raa) - 0.75*({cos(-60)},{sin(-60)})$) node [midway, left] {$2w$};
\draw[rounded corners,black,fill=black,fill opacity=0.2,rotate around={30:(R)}] ($ (R) - \wheelL*({cos(0)},{0.2*sin(90)})$) rectangle ($ (R) + \wheelL*({cos(0)},{0.2*sin(90)})$);
\draw[rounded corners,black,fill=black,fill opacity=0.2,rotate around={60:(F)}] ($ (F) - \wheelL*({cos(0)},{0.2*sin(90)})$) rectangle ($ (F) + \wheelL*({cos(0)},{0.2*sin(90)})$);
\fill (C) circle [radius=2pt];
\draw[color=black!50] (C) -- ($ (C) + 1.55*({cos(0)},{sin(0)})$);
\draw[->,color=red!50,>=latex',ultra thick] (C) -- ($ (C) + 1*({cos(30)},{sin(30)})$);
\node[color=black] (a) at ($ (C) + ({0.44*cos(30)},{0.45*sin(90)})$) {$v$};
\fill[color=blue!70] ($ (C) + \lineRhoSTop*({cos(90+\psis)},{sin(90+\psis)})$) circle [radius=1.5pt];
\draw[color=blue,->] ($ (C) + \lineRhoSTop*({cos(90+\psis)},{sin(90+\psis)})$) -- ($ (C) + \lineRhoSBottom*({cos(-(90-\psis))},{sin(-(90-\psis))})$) node [pos=0.11, left=0.0] {$\rho_s(s)$} node [pos=0.97, left=0.] {$s$};
\draw[color=black!50,->] ($ (C) + \lineRhoSTop*({cos(90+\psis)},{sin(90+\psis)}) + {\lineRhoSTop+\lineRhoSBottom}*({cos(-93)},{sin(-93)})$) arc (-93:-40:{\lineRhoSTop+\lineRhoSBottom});
\node[color=black!50] (a) at ($ (C) + (3.6,-0.25)$) {road centerline};
\draw[-,color=black!50, dashed] ($ (C) + \lineRhoSBottom*({cos(-(90-\psis))},{sin(-(90-\psis))}) - 2*({cos(\psis)},{sin(\psis)}) $) -- ($ (C) + \lineRhoSBottom*({cos(-(90-\psis))},{sin(-(90-\psis))}) + \lineRhoSTop*({cos(\psis)},{sin(\psis)}) $) node [pos=0.95, right=0.1] {tangent};
\draw[-,color=black!50, dashed] ($ (C) - 3*({cos(\psis)},{sin(\psis)}) $) -- ($ (C) + 3*({cos(\psis)},{sin(\psis)}) $);
\draw[-,color=black!50, dashed] ($ (F)$) -- ($ (F) + 1*({cos(30)},{sin(30)}) $);
\draw[-,color=black!50, dashed] ($ (F)$) -- ($ (F) + 1*({cos(60)},{sin(60)}) $);
\draw[->, color=blue] ($ (F) + 0.61*({cos(30)},{sin(30)})$) arc (30:60:0.65);
\node[color=blue] (a) at ($ (F) + 0.77*({cos(45)},{sin(45)})$) {$\delta$};
\draw[->, color=blue] ($ (C) + 1.25*({cos(\psis)},{sin(\psis)})$) arc (\psis:30:1.25);
\node[color=blue] (a) at ($ (C) + 1.5*({cos(22)},{sin(22)})$) {$e_\psi$};
\draw[->, color=blue] ($ (C) + 1.25*({cos(0)},{sin(0)})$) arc (0:\psis:1.25);
\node[color=blue] (a) at ($ (C) + 1.5*({cos(10)},{sin(10)})$) {$\psi_s$};
\draw[<-,color=blue] ($ (C) + 1.7*({cos(\psis)},{sin(\psis)}) $) -- ($ (C) + \lineRhoSBottom*({cos(-(90-\psis))},{sin(-(90-\psis))}) + 1.7*({cos(\psis)},{sin(\psis)}) $) node [pos=0.62, left=-0.1, color=blue] {$e_y$};
%
%
\end{tikzpicture}
\caption{A nonlinear dynamic bicycle model, including the representation of the curvilinear (road-aligned) coordinate system, and vehicle dimensions.}
\label{fig:spatial_bicycle_mdl}
\end{figure}
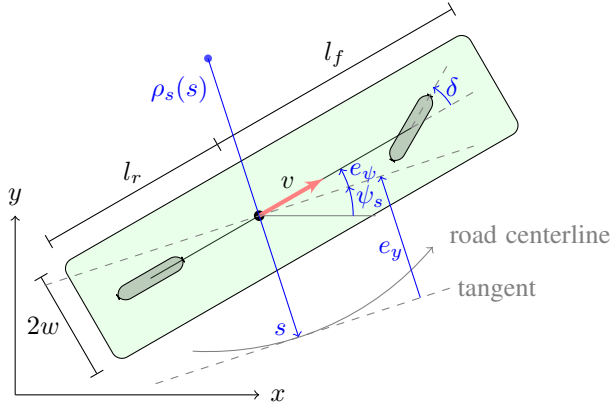

\section{Relating Time to Spatial-based Dynamics Modeling for Mission Planing\label{sec_Main}}

\subsection{Notation and spatial-based vehicle dynamics modeling\label{subsec_spaitalVehDyn}}

We adopt the notation of \cite{plessen2017spatial} and briefly summarize it. Consider a \textit{global} and a \textit{road-aligned} coordinate frame within the $(x,y)$- and $(s,e_y)$-plane, respectively. See Fig. \ref{fig:spatial_bicycle_mdl} for visualization. Coordinates can be transformed by projecting point masses to piecewise-affine (PWA) road line-segments. We differentiate between road centerline\footnote{Throughout, we use the term \textit{road centerline}, noting, however, that such a line may also refer to the centerline of one \textit{lane} of a multi-lane road.} coordinate $s\geq 0$, and actual path length coordinate $\eta\geq 0$ indicating the distance traveled by the vehicle. Note that $\eta \neq s$, unless the vehicle is traveling \textit{perfectly} along the road centerline. A trajetory planned at time $t$ can be defined as
\begin{equation}
z(s) = \begin{bmatrix} e_{\psi}(s) & e_{y}(s) \end{bmatrix}^T, \label{eq_def_mathcalT_etan}
\end{equation} 
with discrete $s\in[s_\tau, s_\tau+S]$, where $S>0$ denotes the road corridor length and $s_\tau$ the vehicle's location at time $\tau$. The equivalent trajectory in the $(x,y)$-plane is defined by $
\mathcal{X}(s) = \begin{bmatrix} x(s) & y(s) & \psi(s) \end{bmatrix}^T$. The classic nonlinear kinematic bicycle model \cite{rajamani2011vehicle} is
\begin{equation}
\begin{bmatrix} \dot{x} & \dot{y} & \dot{\psi} \end{bmatrix}^T = \begin{bmatrix} v\cos(\psi) & v\sin(\psi) & \frac{v}{l}\tan(\delta) \end{bmatrix}^T,\label{eq_dotxypsi_nonlinkinbicmdl}
\end{equation}
assuming the center of gravity (CoG) to be located at the rear axle and $l$ denoting the wheelbase. The front-axle steering angle $\delta$ and vehicle velocity $v$ are used as our control variables. Let us denote time and spatial derivatives by $\dot{x}=\frac{d x}{dt}$ and $x'=\frac{d x}{d s}$, respectively. The spatial equivalent of \eqref{eq_dotxypsi_nonlinkinbicmdl} is derived as follows \cite{graf2017spatial}. By $\dot{e}_\psi=\dot{\psi}-\dot{\psi}_s$, $\dot{e}_y=v \sin(e_\psi)$, $\dot{s}=\frac{\rho_s v \cos(e_\psi)}{\rho_s-e_y}$, 
according to Fig. \ref{fig_ProblVisualization}, 
and expressing $e_\psi'=\frac{\dot{e}_\psi}{\dot{s}}$ and $e_y' = \frac{\dot{e}_y}{\dot{s}}$, we obtain
\begin{equation}
\begin{bmatrix} e_\psi' & e_y' \end{bmatrix}^T \!=\! \begin{bmatrix} \frac{(\rho_s - e_y)\tan(\delta)}{\rho_s l \cos(e_\psi)} -\psi_s' & \frac{\rho_s-e_y}{\rho_s} \tan(e_\psi)  \end{bmatrix}^T.\label{eq_nonlinmdl_steering}
\end{equation}
Importantly, \eqref{eq_nonlinmdl_steering} is entirely independent of vehicle speed $v$. This is characteristic for \textit{kinematic} models, but not the case for \textit{dynamic} vehicle models \cite{graf2017spatial}. We relate control $v$ to our spatial-based modeling as discussed in Section \ref{subsec_relatingTime}. To summarize, we define spatially-dependent state  and control vectors as $z=\begin{bmatrix}e_\psi & e_y \end{bmatrix}^T$ and $u = \begin{bmatrix} v & \delta \end{bmatrix}^T$, respectively. We abbreviate \eqref{eq_nonlinmdl_steering} by $z' = f(z,u)$. Let a discretization grid along the road centerline be defined by $\{s_j\}_{j=0}^N = \{s_0,s_1,\dots,s_N\}$, where $s_j$ is abbreviated for $s_{\tau+j}$ when planning at time $\tau$. Thus, we also have $s_N = s_\tau + S$. The discretization grid is initialized as uniformly spaced. Additional grid points are added such that all (potentially safety margin-inflated) obstacle corners within the $(s,e_y)$-frame are taken into account. For a given set of references $\{e_{\psi,j}^\text{ref}\}_{j=0}^N$, $\{e_{y,j}^\text{ref}\}_{j=0}^N$ and $\{u_j^\text{ref}\}_{j=0}^{N-1}$, the linearized and discretized state dynamics of \eqref{eq_nonlinmdl_steering} are denoted by $z_{j+1} = A_jz_j + B_j u_j + g_j$. For the remainder of this paper, index $j$ refers to the \textit{spatial} discretization grid.

\subsection{Relating time}\label{subsec_relatingTime}

We seek to express time $t$ as a function of space coordinate $s$. We therefore relate $t'=\frac{1}{\dot{s}}$ and obtain
\begin{equation}
t' = \frac{\rho_s - e_y}{\rho_s v \cos(e_\psi)}.\label{eq_def_t_apostrophe}
\end{equation}
Throughout the following, we assume $e_\psi\in(-\frac{\pi}{2},\frac{\pi}{2})$, $\rho_s > e_y$ and $v\geq 0$. This is made to avoid poles at $e_\psi=\pm \frac{\pi}{2}$ in \eqref{eq_nonlinmdl_steering} and \eqref{eq_def_t_apostrophe}. Let the linearization of \eqref{eq_def_t_apostrophe} along the discretization grid be denoted by $t_{j+1}=t_j + a_{t,j} z_j + b_{t,j}u_j + g_{t,j}$.
\begin{property}\label{def_property1}
For timely mission planning (scheduling) and the formulation of convex optimization problems with linear constraints, we require $t_{j+l}> t_j,~\forall l > 0$.   
\end{property}
\begin{proposition}\label{def_proposition1}
The linearized dynamics of \eqref{eq_def_t_apostrophe}, $t_{j+1}=t_j + a_{t,j} z_j + b_{t,j}u_j + g_{t,j}$, violate Property \ref{def_property1}. 
\end{proposition}
\begin{proof}
We express $t_{j+1} = t_j + a_{t,j}^{e_\psi}(e_{\psi,j} - e_{\psi,j}^\text{ref}) + a_{t,j}^{e_y}(e_{y,j} - e_{y,j}^\text{ref}) + b_{t,j}^v(v_j - v_j^\text{ref}) + f_j^\text{ref}$ with $f_j^\text{ref}= f(z_j^\text{ref},u_j^\text{ref})$. We have $b_{t,j}^v = - \frac{\rho_j^\text{ref} - e_{y,j}^\text{ref}}{ \rho_j^\text{ref} (v_j^\text{ref})^2 \cos(e_{\psi,j}^\text{ref}) }$, which \textit{always} is negative. Then, $f_j^\text{ref} + b_{t,j}^{v}(v_j - v_j^\text{ref}) = \frac{\rho_{s,j}^\text{ref} - e_{y,j}^\text{ref} }{ \rho_{s,j}^\text{ref} v_j^\text{ref} \cos(e_{\psi,j}^\text{ref} ) }(2- \frac{v_j}{v_j^\text{ref}})$. To prove the proposition, it suffices to find one counterexample violating $t_{j+1}>t_j$. Such can be constructed by $e_{\psi,j}=e_{\psi,j}^\text{ref}$, $e_{y,j}=e_{y,j}^\text{ref}$ and $v_j > 2v_j^\text{ref}$, which concludes the proof.
\end{proof}
\begin{remark}\label{def_rmk1}
Proposition \ref{def_proposition1} is also valid when simplifying \eqref{eq_def_t_apostrophe} to $t' = 1/v\cos(e_\psi)$ under the assumption of $\rho_s\gg e_y$. It is further also valid for a small $e_\psi$-angle approximation with $\cos(e_\psi)\approx 1 - \frac{e_\psi^2}{2}$. This is since counterexamples can be constructed similarly to before. 
\end{remark}
Proposition \ref{def_proposition1} and Remark \ref{def_rmk1} have important implications. A minimization problem with objective $\min \{t_N\}$ can result in optimal solution $\min \{t_N\}<0$, thereby decoupling variable $t$ from physical interpretation and thus making it useless for time scheduling tasks.
\begin{proposition}\label{def_proposition2}
The coordinate transformation $q_j = \frac{1}{v_j}$ and dynamics
\begin{equation}
t_{j+1} = t_j + \frac{ (s_{j+1} - s_j)(\rho_{s,j}^\text{ref}-e_{y,j}^\text{ref}) }{\rho_{s,j}^\text{ref} \cos(e_{\psi,j}^\text{ref}) } q_j,
\label{eq_tjp1_uj}
\end{equation}
approximate the linearization of \eqref{eq_def_t_apostrophe} for finite reference velocity and satisfy Property \ref{def_property1}.
\end{proposition}
\begin{proof}
Under the coordinate transformation and abbreviating $D_{s,j} = s_{j+1}-s_j$, the discrete equation of \eqref{eq_def_t_apostrophe} is $t_{j+1} = t_j + D_{s,j} \frac{\rho_{s,j}-e_{y,j}}{\rho_{s,j}\cos(e_{\psi,j})}q_j$, and linearized $t_{j+1} = t_j + a_{t,j}^{e_\psi,q}(e_{\psi,j}-e_{\psi,j}^\text{ref}) + a_{t,j}^{e_y,q}(e_{y,j}-e_{y,j}^\text{ref}) + b_{t,j}^q q_j$. Note that $a_{t,j}^{e_\psi,q}$ and $a_{t,j}^{e_y,q}$ are proportional to $q_j^\text{ref}$ and can be positive- \textit{or} negative. Thus, to guarantee $t_{j+1} > t_j$, they must be eliminated from $t_{j+1}$-dynamics. This is achieved for $q_j^\text{ref}\rightarrow 0$ (i.e., $v_j^\text{ref}\rightarrow \infty$) by proportionality. The remainder yields \eqref{eq_tjp1_uj}, which is the exact linearization of \eqref{eq_def_t_apostrophe} for $v_j^\text{ref}\rightarrow \infty$, and approximate for finite reference velocities.
\end{proof}
Importantly, the coordinate transformation does not affect the linearized and discretized state dynamics for $z_{j+1}$. This is since they are independent of $v_j$, and consequently also independent of $q_j$. This is in contrast to control rate constraints, as discussed next. For simplicity, we continue to denote the new control vector by $u$, i.e., $u=\begin{bmatrix} q & \delta \end{bmatrix}^T$.

\subsection{Control rate constraints}

Continuous rate constraints for control variables $v$ and $\delta$ are of the form 
\begin{equation}
\dot{v}^\text{min} \leq \dot{v} \leq \dot{v}^\text{max} \quad \text{and} \quad  \dot{\delta}^\text{min} \leq \dot{\delta} \leq \dot{\delta}^\text{max},
\label{eq_time_rateCstrts}
\end{equation}
whereby the bounds ($\dot{v}^\text{min}$, $\dot{v}^\text{max}$, $\dot{\delta}^\text{min}$ and $\dot{\delta}^\text{max}$), in general, are nonlinear functions of the vehicle's operating point. They are time-varying parameters, for example, dependent on engine speed and torque. By applying the spatial coordinate transformation, a discretization, the change of variables according to Section \ref{subsec_relatingTime}, and assuming the bounds to remain constant for the duration of the planning horizon, we obtain

\begin{small}
\begin{align}
&  \frac{D_{s,j} \dot{v}^\text{min} (\rho_{s,j} - e_{y,j}) q_j}{\rho_{s,j} \cos(e_{\psi,j})} \leq \frac{1}{q_{j+1}} - \frac{1}{q_j} \leq  \frac{D_{s,j} \dot{v}^\text{max} (\rho_{s,j} - e_{y,j}) q_j}{\rho_{s,j} \cos(e_{\psi,j})},\label{eq_deltaq_rateCstrts}\\
&  \frac{D_{s,j} \dot{\delta}^\text{min} (\rho_{s,j} - e_{y,j}) q_j}{\rho_{s,j} \cos(e_{\psi,j})} \leq \delta_{j+1} - \delta_j \leq \frac{ D_{s,j} \dot{\delta}^\text{max} (\rho_{s,j} - e_{y,j}) q_j}{\rho_{s,j} \cos(e_{\psi,j})},\label{eq_deltad_rateCstrts}
\end{align}
\end{small}\normalsize
whereby we abbreviated $D_{s,j}=s_{j+1} - s_j$. Thus, linear and control channel-separated rate constraints in \eqref{eq_time_rateCstrts} are rendered not only nonlinear, but additionaly state-dependent, and also velocity-dependent for steering control. This has two implications. First, to formulate linearly constrained optimization problems, we require the linearization of \eqref{eq_deltaq_rateCstrts} and \eqref{eq_deltad_rateCstrts}. Dependent on the quality of underlying reference trajectories, this may incur significant distortions. Second, while the discrete form of \eqref{eq_time_rateCstrts} can always be guaranteed to be feasible (assuming a feasible initialization), for \eqref{eq_deltaq_rateCstrts} and \eqref{eq_deltad_rateCstrts} this is not the case anymore. Thus, slack variables are required. Let us consider two dregrees of simplification of \eqref{eq_deltaq_rateCstrts} and \eqref{eq_deltad_rateCstrts}. First, we assume large $\rho_{s,j}$ and small $e_{\psi,j}$, and consequently approximate
$\frac{(\rho_{s,j} - e_{y,j})}{\rho_{s,j} \cos(e_{\psi,j})} \approx 1$, thereby rendering \eqref{eq_deltaq_rateCstrts} and \eqref{eq_deltad_rateCstrts} state-independent, but maintaining velocity-dependent bounds. Thus, steering rate constraints still depend on $q_j$. Second, we additionally eliminate this velocity-dependency and formulate 
\begin{align}
&  \tilde{T} \dot{v}^\text{min} \leq \frac{1}{q_{j+1}} - \frac{1}{q_j} \leq \tilde{T} \dot{v}^\text{max},~j=0,\dots,N-2,\label{eq_deltaq_rateCstrts_III}\\
&  \tilde{T} \dot{\delta}^\text{min}  \leq \delta_{j+1} - \delta_j \leq \tilde{T} \dot{\delta}^\text{max},~j=0,\dots,N-2,\label{eq_deltad_rateCstrts_III}
\end{align}
whereby, in general, $\tilde{T}$ is a parameter choice. It is typically selected as the sampling time $T_s$ in a closed-loop MPC-setting. For spatial-based predictive control, $T_s$ must be related to the spatial discretization grid \cite{graf2017spatial}. Formulations \eqref{eq_deltaq_rateCstrts_III} and \eqref{eq_deltad_rateCstrts_III} bear the advantage of separating control channels and are therefore our preferred form for \textit{spatial-based} rate constraints. In practice, we employ bounds that are constant over the spatial planning horizon, but time-varying in a closed-loop MPC-setting and dependent on the vehicle's operating point. We denote the linearization of \eqref{eq_deltaq_rateCstrts_III} by 
\begin{equation}
c_{q,j}^\text{min} \leq b_{q,j+1}q_{j+1} + b_{q,j} q_j \leq c_{q,j}^\text{max},~j=0,\dots,N-2.\label{eq_deltaq_rateCstrts_IIIlinear}
\end{equation}
To summarize this section, the transformation of time-dependent control rate constraints \eqref{eq_time_rateCstrts} to the road-aligned coordinate frame is not trivial and to be considered as the main disadvantage of a spatial-based system representation. We opted for the ``simple'' control channel-seperating forms \eqref{eq_deltad_rateCstrts_III} and \eqref{eq_deltaq_rateCstrts_IIIlinear} in linearly constrained optimization problems, and discussed the role of $\tilde{T}$ as a transformation parameter.

\subsection{Vehicle dimension constraint heuristic\label{subsec_VDCH}}

For the navigation of automated vehicles, especially in very constrained environments, vehicle dimensions must be accounted for. This is particularly relevant for large-sized vehicles such as heavy-duty trucks or buses. A spatial-based planning method, that is based on iterative linearization of nonlinear vehicle dimension constraints, is presented in \cite{plessen2017spatial}. Here, a different approach is taken. A velocity-dependent and computationally much less demanding \textit{heuristic} is proposed. For clarification, the entire discussion of relating time to spatial-based system modeling, that is essential for the simultaneous planning of \textit{both} steering and velocity trajectories when using kinematic vehicle models in the spatial domain, is also absent in \cite{plessen2017spatial}.

We denote road corridor constraints as 
\begin{equation}
e_{y,j}^\text{min} + \Delta e_{y,\tau} \leq e_{y,j} \leq  e_{y,j}^\text{max} - \Delta e_{y,\tau}, ~j=1,\dots,N,\label{eq_def_Deltaey_VDC_corridor}
\end{equation}
with corridor boundaries $e_{y,j}^\text{min}$ and $e_{y,j}^\text{max}$, and margin $\Delta e_{y,\tau}\geq 0$ determined at time $\tau$. For point mass trajectory planning without vehicle dimension constraints we have $\Delta e_{y,\tau}=0$. 
\begin{proposition}\label{prop_Deltaey_VDC}
To guarantee safe vehicle operation within road boundaries according to \eqref{eq_def_Deltaey_VDC_corridor}, assuming rectangular vehicle dimensions with $w<l_f$, and assuming forward motion, we require $\Delta e_{y,\tau} = \max_{e_\psi\in\mathcal{E}} \Delta e_{y,\tau}(e_\psi)$ with
\begin{equation}
\Delta e_{y,\tau}(e_\psi)  = l_f \sin(e_\psi) + w \cos(e_\psi),\label{eq_max_Deltaeytau}
\end{equation}  
and $\mathcal{E}=(-\pi/2,\pi/2)$.
\end{proposition}
\begin{proof} From the definitions of $l_f$ and $w$ in Fig. \ref{fig:spatial_bicycle_mdl}. 
\end{proof}
\begin{remark}
The selection of $\Delta e_{y,\tau}$ according to Proposition \ref{prop_Deltaey_VDC} is conservative. Let us denote the associated angle by $e_\psi^\text{max}$, whereby $e_\psi^\text{max}=\tan^{-1}(l_f/w)$ is derived from the maximization. For $w=0.9$ and $l_f=3.5$ used in simulations, we obtain $e_\psi^\text{max}=75.6^\circ$. Such a large deviation from the road heading is not admissible at high speeds. By tighter constraining $\mathcal{E}$, the level of conservativeness can be reduced. Note that we have $w\leq \Delta e_{y,\tau}(e_\psi) \leq \Delta e_{y,\tau}(e_\psi^\text{max})$ for $e_\psi\in[0,e_\psi^\text{max}]$, and that $\Delta e_{y,\tau}(e_\psi)$ is strictly monotonously increasing for that heading range. 
\end{remark}
In the following, we derive a heuristic for the selection of $\Delta e_{y,\tau}$. The ideas are a) to relate to vehicle  speed, and b) to obtain monotonously increasing $\Delta e_{y,\tau}$ with increasing $v$. Let $v^\text{max}$ denote the maximum highway speed limit (e.g., $v^\text{max}=120$km/h), and $v_\tau$ the traveling velocity at time $\tau$. Let us first state the algorithm before discussing two variants: 
\begin{enumerate}
\item Determine an angle $e_\psi^v\in[0,e_\psi^\text{max}]$. 
\item Compute $e_\psi=\frac{v_\tau}{v^\text{max}} e_\psi^v$.
\item Compute $\Delta e_{y,\tau} = l_f \sin(e_\psi) + w \cos(e_\psi)$ and adapt corridors in \eqref{eq_def_Deltaey_VDC_corridor}.
\end{enumerate}
Note that by design, both of the mentioned motivating ideas are addressed. Also, the level of conservativeness can be controlled by $e_\psi^v$. We considered two options. First, $e_\psi^v = e_\psi^\text{max}$. Second, we note that specific vehicle velocities only admit a very limited deviation from the road heading direction to take corrective steering action within a limited ``reaction time'' $\Gamma$. Assuming a road boundary $\tilde{e}_y^\text{max}>0$ and a vehicle position $\tilde{e}_{y} < \tilde{e}_y^\text{max}$,  the front left vehicle corner reaching the road boundary $\tilde{e}_y^\text{max}$ within time $\Gamma$ can be expressed as
\begin{equation}
\tilde{e}_{y} + \Gamma v_\tau \sin(\tilde{e}_{\psi}) + l_f \sin(\tilde{e}_\psi) + w \cos(\tilde{e}_\psi) = \tilde{e}_y^\text{max}, \label{eq_ReactionTime}
\end{equation}
where $\tilde{e}_\psi\geq 0$ denotes the vehicle heading. Thus, for our second variant, we solve \eqref{eq_ReactionTime} analytically for $\tilde{e}_\psi$, and set $e_\psi^v = \tilde{e}_\psi$ in step 1). This method is less conservative/generates smaller $\Delta e_{y,\tau}$ than the first variant with $e_\psi^v = e_\psi^\text{max}$. This is since $\tilde{e}_\psi<e_\psi^\text{max}$ for typical parameter choices, and because $\Delta e_{y,\tau}(e_\psi)$ is strictly monotonously increasing for an $e_\psi\in[0,e_\psi^\text{max}]$. Note that the front left vehicle corner was considered for derivation. This is appropriate for our parameter choices of $\tilde{e}_y=|e_y(s_\tau)|$ and $\tilde{e}_y^\text{max}=\min_{j} \{\min(e_{y,j}^\text{max},|e_{y,j}^\text{min}|)\}_{j=1}^N$, thereby accounting for \textit{both} road boundaries. In simulations of Section \ref{sec_numerical_experiments}, we assumed $l_f=3.5$, $w=0.9$, and set parameter $\Gamma=0.05$ for racing performance (time-optimal ``cutting'' of curves). For increased safety/larger $\Delta e_{y,\tau}$, $\Gamma$ must be further decreased.

\subsection{Linear programming\label{subsec_OPformulation}}

We propose the following linear programming (LP):
\small
\begin{subequations}
\label{eq:LP}
\begin{align}
\min &\ \  t_N+ \max|\delta| +  \max|D_1 \delta| +  W_\sigma \sum\nolimits_{i=1}^4 \sigma_i \label{eq_OP_objFcn}\\
\mathrm{s.t.} &\ \ z_0 = z(s_\tau),~t_0=\tau,~u_{-1}=u(s_\tau-D_s) \label{eq_OP_1stcstrt}\\
&\ \ z_j = \begin{bmatrix} e_{\psi,j} & e_{y,j}  \end{bmatrix}^T, \ j =0,\dots,N, \\
&\ \ u_j = \begin{bmatrix} q_j & \delta_j  \end{bmatrix}^T, \ j =0,\dots,N-1, \\
&\ \ z_{j+1} = A_j z_j + B_ju_j + g_j, \ j =0,\dots,N-1, \label{eq:OP_zjp1_eq}\\
&\ \ t_{j+1} = t_j + \frac{ D_{s,j}(\rho_{s,j}^\text{ref}-e_{y,j}^\text{ref}) q_j}{\rho_{s,j}^\text{ref} \cos(e_{\psi,j}^\text{ref}) }, \ j=0,\dots,N-1, \label{eq:OP_tjp1_eq}\\
&\ \ e_{\psi}(s_\tau+S) - \sigma_1 \leq e_{\psi,N} \leq  e_{\psi}(s_\tau+S) + \sigma_1, \\
&\ \ e_{y}(s_\tau+S) - \sigma_2 \leq e_{y,N} \leq  e_{y}(s_\tau+S) + \sigma_2, \\
&\ \ e_{y,j}^\text{min} + \Delta e_{y,\tau} - \sigma_3 \leq e_{y,j} \leq  e_{y,j}^\text{max} - \Delta e_{y,\tau} + \sigma_3,\notag\\[-1pt]
& \hspace{5cm} \ j = 1,\dots,N, \label{eq:OP_ey_constrts}\\
& t_j^\text{WP} - \sigma_4 \leq t_j \leq t_j^\text{WP} + \sigma_4,~z_j\in\mathcal{Z}_j^\text{WP}, \ \forall j\in\mathcal{J}^\text{WP},\label{eq:OP_tWP_cstrts}\\
& c_{q,j}^\text{min} \leq b_{q,j+1}q_{j+1} + b_{q,j} q_j \leq c_{q,j}^\text{max},~j=0,\dots,N-2,\label{eq:OP_deltaq_cstrts}\\
& \tilde{T} \dot{\delta}^\text{min}  \leq \delta_{j+1} - \delta_j \leq \tilde{T} \dot{\delta}^\text{max},~j=0,\dots,N-2,\\
&\ \ \frac{1}{v^\text{max}} \leq q_j \leq  \frac{1}{v^\text{min}}, \ j=0,\dots,N-1, \label{eq:OP_qj_hardcstrts}\\
&\ \ \delta^\text{min} \leq \delta_j \leq  \delta^\text{max}, \ j=0,\dots,N-1, \\
& \ \ \sigma_1\geq 0,~\sigma_2\geq 0,~\sigma_3\geq 0,~\sigma_4\geq 0, 
\end{align}
\end{subequations} \normalsize
with decision variables $\{u_j\}_{j=0}^{N-1}$, $\{\sigma_i\}_{i=1}^4$ and optimization horizon $N$. The absolute value is denoted by $|\cdot|$ and $D_1$ indicates the spatial-based first-order difference operator acting on vectorized steering angle $\delta \in\mathbb{R}^{N\times 1}$. Objective function \eqref{eq_OP_objFcn} trades-off time-optimality and a minmax-type objective resulting in minimized steering actuation (\textit{smooth steering}). In experiments we selected $W_\sigma=10^4$. It is the only weight in \eqref{eq_OP_objFcn}. Spatiotemporal constraints \eqref{eq:OP_tWP_cstrts} are used for time scheduling. They indicate the times at which \textit{waypoints} (WP) are meant to be traversed. We define $\mathcal{J}^\text{WP} = \{j:s_j=s_j^\text{WP},~s_j^\text{WP}\in\mathcal{S}^\text{WP},~j=1,\dots,N\}$, where $\mathcal{S}^\text{WP}$ is an input set that may be provided, for example, by a higher-level mission planning algorithm. In addition to $\mathcal{S}^\text{WP}$, such algorithm must provide the corresponding scheduling times $\mathcal{T}^\text{WP} = \{t_j^\text{WP},~\forall j\in\mathcal{J}^\text{WP}\}$. The states in which the waypoints are reached is constrainted by $\mathcal{Z}_j^\text{WP}$.  Hard constraints \eqref{eq:OP_qj_hardcstrts} are derived from $v^\text{min} \leq v_j \leq v^\text{max},~j=0,\dots,N-1$, and from the coordinate transformation according to Section \ref{subsec_relatingTime}, whereby $v^\text{max}$ denotes the road speed limit and $v^\text{min}\geq 0$ the minimum permissible velocity. Since \eqref{eq:OP_zjp1_eq}, \eqref{eq:OP_tjp1_eq} and \eqref{eq:OP_deltaq_cstrts} depend on reference trajectories, \eqref{eq:LP} is solved \textit{twice}, see Section \ref{subsec_IncorpFrictCstrts}. This often improved results since time-optimal trajectories typically exploit the complete road width, and therefore incur lateral deviations from the road centerline. For the first iteration, we initialize state trajectories along the road centerline, i.e., $e_{\psi,j}^\text{ref}=0$ and $e_{y,j}^\text{ref}=0,~\forall j$, and select $q_j^\text{ref}=\frac{1}{v_\tau}$ and $\delta_j^\text{ref}=0,~\forall j$. Finally, moving obstacles are accounted for by their \textit{velocity}- and \textit{trajectory}-adjusted mappings to the road-aligned coordinate frame according to the method of \cite[Sect. III-E]{graf2017spatial}. 

\subsection{Incorporating friction constraints and summary of TOSS \label{subsec_IncorpFrictCstrts}}

LP \eqref{eq:LP} does not yet incorporate friction constraints. Two options were considered. The first was motivated by \cite{mcnaughton2011motion}, where a term penalizing lateral accelerations was added to the cost function in order for the vehicle to automatically slow down in anticipation of tight turns. We approximated lateral acceleration as $a_y=v\dot{\psi}=\frac{v^2 \tan(\delta)}{l}$ \cite{werling2012automatic}, conducted a linearization, and incorporated a minmax-type penalty in the cost function. This method has three (significant) disadvantages: references $v_j^\text{ref}$ and $\delta_j^\text{ref},~\forall j=0,\dots,N$ are required for computation, it is not obvious how to \textit{weight} said penalities, and the formulation does \textit{not} guarantee operation within friction limits. We therefore instead opt for the following algorithm which we label TOSS (time-optimal smooth steering): 
\begin{enumerate}
\item Solve \eqref{eq:LP} to obtain $\{e_{\psi,j}\}_{j=0}^N$, $\{e_{y,j}\}_{j=0}^N$, $\{v_j=1/q_j\}_{j=0}^{N-1}$, $\{\delta_j\}_{j=0}^{N-1}$, $\{t_j\}_{j=0}^{N}$ along $\{s_j\}_{j=0}^N$.
\item Use the trajectories of Step 1) as references for a second solution of \eqref{eq:LP} with \textit{additional} constraints
\begin{equation}
q_j \geq \frac{1}{v_j^\text{max,fric}},~\forall j=0,\dots,N-1,
\end{equation}
where $v_j^\text{max,fric}$ is computed according to \cite[Sect. 3]{funke2016simple}, assuming a friction coefficient $\mu$ (in Section \ref{sec_numerical_experiments}, $\mu=0.8$), and denoting the maximum admissible velocity permitting operation within vehicle tire friction limits.  
\end{enumerate}
The first step is to generate a suitable vehicle trajectory based on which $v_j^\text{max,fric}$ can be computed. The second step is to refine velocity control; see Section \ref{sec_numerical_experiments} for the implications. 

\subsection{Extension to dynamic vehicle models\label{subsec_DynVehMdl}}

So far, the discussion focused on \textit{kinematic} vehicle models. Let us extend the discussion to \textit{dynamic} vehicle models \cite{kong2015kinematic}, \cite{graf2017spatial}. The equivalent to \eqref{eq_def_t_apostrophe} can be derived analogously, resulting in 
\begin{equation}
t' = \frac{\rho_s - e_y}{\rho_s \left(v_x \cos(e_\psi) - v_y \sin(e_\psi)\right)},\label{eq_def_t_apostrophe_dynamic}
\end{equation}
where $v_x$ and $v_y$ denote longitudinal and lateral velocities relative to the inertial vehicle frame. Importantly, $v_x$ and $v_y$ are vehicle states and not control variables anymore.
\begin{proposition}
The linearized dynamics of \eqref{eq_def_t_apostrophe_dynamic}, of the form 
$t_{j+1}=t_j + a_{t,j} z_j + b_{t,j}u_j + g_{t,j}$, violate Property \ref{def_property1}.
\end{proposition}
\begin{proof}
Similar to the one of Proposition \ref{def_proposition1}.
\end{proof}
Additional complexity arises from state variable $v_y$, that is absent in the kinematic case. In practice, typically $v_y\ll v_x$.
\begin{proposition}
The coordinate transformation $q_j^{v_x} = \frac{1}{v_{x,j}}$ and dynamics
\begin{equation}
t_{j+1} = t_j + \frac{ (s_{j+1} - s_j)(\rho_{s,j}^\text{ref}-e_{y,j}^\text{ref}) }{\rho_{s,j}^\text{ref} \cos(e_{\psi,j}^\text{ref}) } q_j^{v_x},
\label{eq_tjp1_uj_dynamicMdl}
\end{equation}
approximate the linearization of \eqref{eq_def_t_apostrophe} and satisfy Property \ref{def_property1}.
\end{proposition}
\begin{proof}
Similar to the one of Proposition \ref{def_proposition2}. 
\end{proof}
The time dynamics \eqref{eq_tjp1_uj_dynamicMdl} are characteristic for the dynamic vehicle model case in the sense that state $v_y$ is omitted entirely from consideration. This is done to comply with Property \ref{def_property1}. Finally, note that according to the coordinate transformation, $q_j^{v_x} = \frac{1}{v_{x,j}}$, also all state equations (and consequently linearization and discretization routines) need to be updated. Based on \eqref{eq_tjp1_uj_dynamicMdl}, a LP similar to \eqref{eq:LP} can now be formulated for the case of dynamic vehicle models.

\subsection{Deployment for time-based reference trajectory tracking\label{subsec_HowUseForTimeBased}}

The trajectory resulting from the solution of \eqref{eq:LP} is given by $\{e_{\psi,j}\}_{j=0}^N$, $\{e_{y,j}\}_{j=0}^N$, $\{v_j=1/q_j\}_{j=0}^{N-1}$, $\{\delta_j\}_{j=0}^{N-1}$ and $\{t_j\}_{j=0}^{N}$, whereby all variables are described along the discretization grid $\{s_j\}_{j=0}^N$. There are various options for deployment. Using the, in general, non-uniformly spaced discrete time $\{t_j\}_{j=0}^{N}$ as the dependent variable, the aforementioned trajectories can be employed as references for a \textit{time-based} tracking controller. Either employing the kinematic bicycle model from Section \ref{subsec_spaitalVehDyn}, and correspondingly
\begin{equation}
\begin{bmatrix} \dot{s} & \dot{e}_\psi & \dot{e}_y \end{bmatrix}^T = \begin{bmatrix} \frac{\rho_s v \cos(e_\psi)}{\rho_s - e_y} & \frac{v\tan(\delta)}{l} - \dot{\psi}_s & v \sin(e_\psi) \end{bmatrix}^T,
\label{eq_def_dotsepsiey}
\end{equation}
or, for example, using a \textit{higher fidelity} model such as a dynamic bicycle model in road-aligned coordinate frame \cite{rajamani2011vehicle}. An alternative method is to solve the TOSS-algorithm according to Section \ref{subsec_IncorpFrictCstrts} at every sampling interval and directly apply $v_0$ and $\delta_0$ to the vehicle's low-level controllers; thereby combining reference path planning, velocity planning and reference tracking in \textit{one} step and in form of a spatial-based receding horizon control (RHC) scheme. Then, additional attention has to be addressed to the relation between sampling intervals and discretization grid \cite{graf2017spatial}.

\section{Numerical Simulations\label{sec_numerical_experiments}}

\subsection{Experiment 1}

The minimum-time traversal of a curvy road segment with one obstacle is sought. Two scheduling constraints are considered: $\mathcal{S}^\text{WP} = \{s^\text{obj},170\}$ and $\mathcal{T}^\text{WP} = \{10,16\}$, where $s^\text{obj}$ denotes the coordinate at which the obstacle is first encountered. For both waypoints, we did not further constrain admissible lateral vehicle position. The results are summarized in Fig. \ref{fig_EX2}. Note that only for better visualization of presented concepts both road lanes were admitted for maneuvering. In practice, the permissible road width can be controlled conveniently by \eqref{eq:OP_ey_constrts}. Several observations can be made. First, the spatiotemporal waypoints are met accurately. See Fig. \ref{fig_EX2} for the resulting optimal velocity trajectory. Second, TOSS satisfies all requirements of Objective \ref{obj_probForm}, combining steering \textit{and} velocity control. Third, Fig. \ref{fig_EX2_vehDimZoom} visualizes the effect of the vehicle dimension constraint discussed in Section \ref{subsec_VDCH}. 

\begin{figure*}[ht]
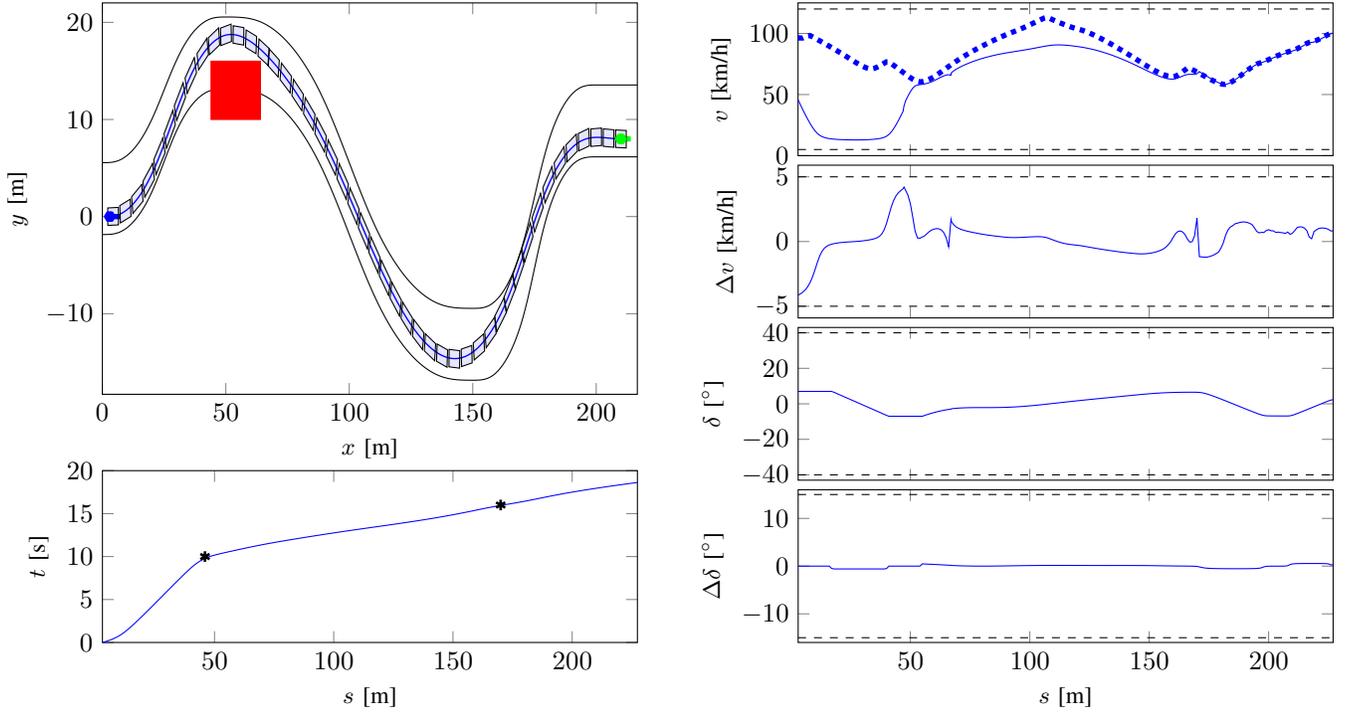

\vspace{0.3cm}
\centering
\input{fig_EX2_traj.tikz}~~~~~
\input{fig_EX2_ctrl.tikz}
\caption{Example 1. State, time and control trajectories. The obstacle is indicated in red. The two scheduling times and corresponding spatial coordinates are visualized by two black asterisks. Control absolute and rate constraints are indicated by black dashed lines. Velocity $v^\text{max,fric}(s)$ is generated according to Section \ref{subsec_IncorpFrictCstrts} and displayed by the blue dotted line. All solutions of TOSS are indicated by blue solid lines.}
\label{fig_EX2}
\end{figure*}

\begin{figure}[ht]
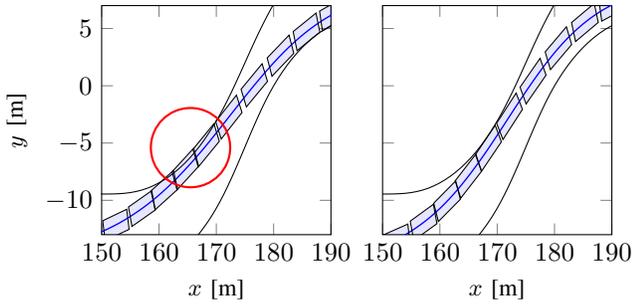

\vspace{0.3cm}
\centering
\input{fig_EX2_zoomWithOutMargin.tikz}
\hspace{-0.2cm}\input{fig_EX2_zoom.tikz}
\caption{Example 1. (Left) Vehicle dimensions are exceeding the road boundary in case of the static selection $\Delta e_{y,\tau}=w$; see the red-circled segment. (Right) Zoom-in into Fig. \ref{fig_EX2}. Road boundary constraints are met tightly due to the adaptive $\Delta e_{y,\tau}$-selection according to Section \ref{subsec_VDCH}.}
\label{fig_EX2_vehDimZoom}
\end{figure}

\begin{figure}[!ht]
\vspace{0.3cm}
\centering
\input{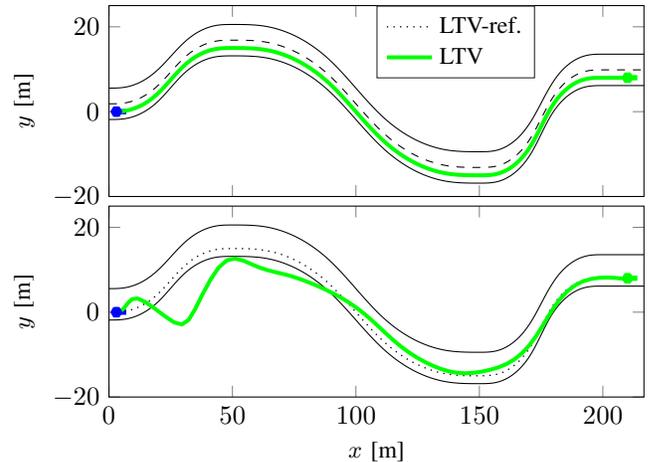}
\caption{Examples 2 (top) and 3 (bottom). For indication of a two-lane road, the separating strip is given by the black dashed line. For clarity, it is omitted from the bottom plot. For both examples, friction constraints \eqref{eq:LTV_vkmaxfric} are \textit{not} considered; see Section \ref{subsec_compTimeSpace} for the discussion. Reference speeds are set as $v^\text{ref}=50$km/h (top) and $v^\text{ref}=120$km/h (bottom), respectively. The reference path (LTV-ref.) can be tracked perfectly only for the former example. Since the reference path is identical in both examples, the lane departure is caused solely by the difference in \textit{reference speed}.}
\label{fig_EX23}
\end{figure}

\subsection{Comparisons with a time-based method\label{subsec_compTimeSpace}}

For comparison to a time-based method, we employ a linear time-varying model predictive control (LTV-MPC) approach \cite{falcone2007linear}, \cite{gutjahr2016lateral}. With \textit{time} serving as the dependent variable, we formulate the quadratic programming (QP):

\small
\begin{subequations}
\label{eq:LTV-MPC}
\begin{align}
\min &\ \  \sum_{k=1}^{K}\|\xi_k - \xi_k^\text{ref} \|_2^2  + \sum_{k=0}^{K-1} \| u_k - u_k^\text{ref} \|_2^2 + \| u_k - u_{k-1} \|_2^2 \label{eq_LTV_objFcn}\\
\mathrm{s.t.} &\ \ \xi_0 = \mathcal{X}(s_\tau),~u_{-1}=u(\tau-T_s), \\
&\ \ \xi_{k+1} = \bar{A}_k \xi_k + \bar{B}_ku_k + \bar{g}_k, \ k =0,\dots,K-1, \label{eq:LTV_xikp1_eq}\\
&\ \ v_k \leq  v_k^\text{max,fric}, \ k=0,\dots,K-1, \label{eq:LTV_vkmaxfric}\\
&\ \ u^\text{min} \leq u_k \leq  u^\text{max}, \ k=0,\dots,K-1, \label{eq:LTV_uk_hardcstrts}\\
&\ \ \Delta u^\text{min} \leq u_k-u_{k-1} \leq  \Delta u^\text{max}, \ k=0,\dots,K-1, \label{eq:LTV_Deltauk_hardcstrts}
\end{align}
\end{subequations} \normalsize 
with optimization variables $\{u_k\}_{k=0}^{K-1}$, $u=\begin{bmatrix} v & \delta \end{bmatrix}^T$, $\|\cdot\|_2$ denoting the 2-norm, optimization horizon $K$, sampling time $T_s$, subscript $k$ indexing sampling times over the optimization horizon, $\xi=\begin{bmatrix} x & y & \psi \end{bmatrix}^T$, and \eqref{eq:LTV_xikp1_eq} indicating the linearized and discretized vehicle dynamics of \eqref{eq_dotxypsi_nonlinkinbicmdl}. In simulations, we assume $T_s=0.1s$. Several remarks can be made. First, the focus of this paper is on trajectory planning. Thus, \eqref{eq:LTV-MPC} is solved once at planning time $\tau$.  We abbreviate the method as LTV in the following (instead of LTV-MPC). Second, \eqref{eq_LTV_objFcn} represents a reference tracking controller. In numerical simulations, we use the centerline coordinates of a lane for state references $\xi_k^\text{ref}$. For control references we select $\delta_k^\text{ref}=0$, and $v_k^\text{ref}$ dependent on the experiment. References are expressed with time as the dependent variable and must therefore be interpolated accounting for \textit{both} $T_s$ and, decisively, the reference speed, $v_k^\text{ref}$, at which the reference path is meant to be traversed. The reason for model selection \eqref{eq_dotxypsi_nonlinkinbicmdl} is that the resulting method according to \eqref{eq:LTV-MPC} performs very robustly when tracking reference paths at constant velocity, even if these paths exhibit discontinuous changes in curvature. Furthermore, it enables to work directly with positioning coordinates, not requiring a coordinate transformation to a road-aligned coordinate frame. Friction constraints are incorporated by computing $v_k^\text{max,fric}$ along the lane centerline~trajectory, and adapting $v_k^\text{ref}=\min(v_k^\text{ref},v_k^\text{max,fric})$. 

Four experiments are reported (Examples 2-5). They are meant to illustrate the benefits of the proposed method TOSS. Rate constraints are synchronized for spatial- and time-based methods with selection $\tilde{T}=T_s$. Throughout, initial vehicle velocity is set as 50km/h. Example 2 assumes a reference speed of 50km/h. In contrast, Example 3 assumes 120km/h and dismisses friction constraints (see the discussion below). The results of both experiments are displayed in Fig. \ref{fig_EX23}. Example 4 seeks time-optimal road traversal considering friction constraints. Results are displayed in Table \ref{tab_Ex4Results} and Fig. \ref{fig_EX4}. The corresponding timings shall be denoted by $t^{\text{LTV},\star}$ and $t^{\text{TOSS},\star}$. Example 5 considers one spatiotemporal constraint for TOSS. The end of the corridor, $s_N$, shall be reached at the time $t_N=t^{\text{LTV},\star}$. The result is summarized in Fig. \ref{fig_EX5}.
\begin{figure}[ht]
\vspace{0.3cm}
\centering
\input{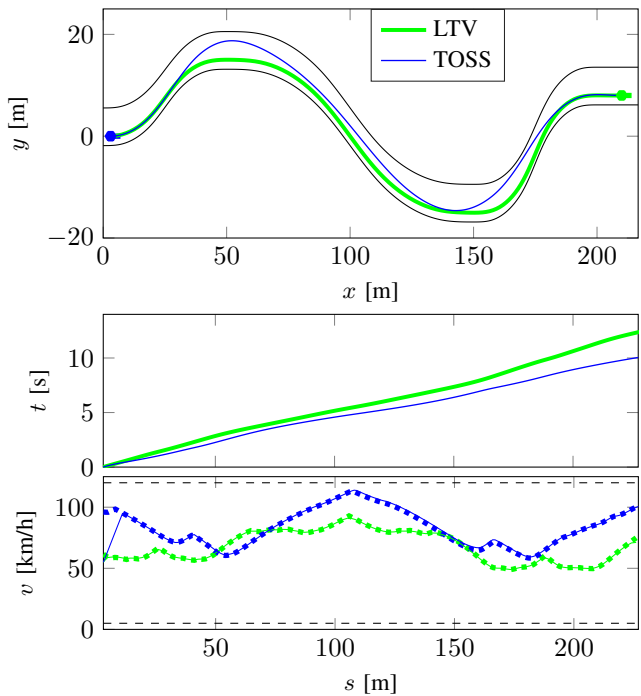}
\caption{Example 4. Comparison of LTV and TOSS. For time-optimal road traversal, both LTV and TOSS return velocity trajectories operating at the vehicle's tire friction limits; $v^\text{max,fric}(s)$ is displayed as the green and blue dotted line for LTV and TOSS, respectively. Because of steering trajectories exploiting the available road width for TOSS, higher $v^\text{max,fric}(s)$ result, that consequently translate to a faster road traversal time, see also Table \ref{tab_Ex4Results}.} \label{fig_EX4}
\end{figure}
\begin{figure}[ht]
\vspace{0.3cm}
\centering
\input{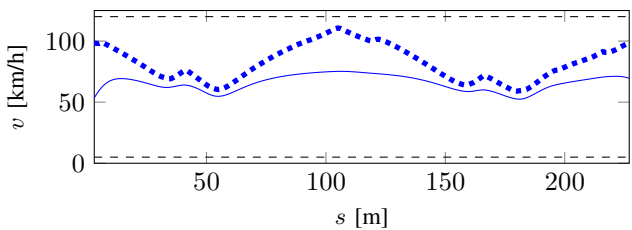}
\caption{Example 5. The velocity profile generated by TOSS when considering one spatiotemproal waypoint: $\mathcal{T}^\text{WP}=\{t^{\text{LTV},\star}$\} and $\mathcal{S}^\text{WP}=\{s_N\}$, where $t^{\text{LTV},\star}=12.4$ is the outcome of Example 4 for LTV. The blue dotted and solid line indicate $v^\text{max,fric}(s)$ and $v(s)$, respectively.} 
\label{fig_EX5}
\end{figure}
\begin{table}[ht]
\begin{center}
\vspace{0.3cm}
\caption{Example 4. The minimum time for road traversal is indicated by $t^\star$. Recorded maximum absolute velocity and velocity rate changes are indicated by $|v|^\text{max}$ and $|\Delta v|^\text{max}$. Similarly for $|v|^\text{min}$, $|\delta|^\text{max}$ and $|\Delta \delta|^\text{max}$.} 
\label{tab_Ex4Results}
\bgroup
\def\arraystretch{1}
\begin{tabular}{|c|c|c|c|c|c|c|}
\hline
\rowcolor[gray]{0.8}  &  $t^\star$ & $|v|^\text{min}$  & $|v|^\text{max}$ & $ |\Delta v|^\text{max}$ & $|\delta|^\text{max}$ & $|\Delta \delta|^\text{max}$ \\ \hline 
LTV &  12.4 & 48.7 &  91.5  &  1.0 &  10.4  &  2.3\\ 
\hline
TOSS &  10.1 & 55.0 &  113.8  &  5.0 &  7.0  &  0.6\\ 
\hline
\end{tabular}
\egroup
\end{center}
\end{table}

\subsection{Interpretations and guidelines for using TOSS}

Examples 2-3 illustrate the importance of reference trajectories for performance of the time-based method of Section \ref{subsec_compTimeSpace}. Even if a very suitable reference path is provided, namely, the obstacle-free lane centerline, safe operation is still not guaranteed. Unsuitable \textit{reference speeds} can cause the vehicle to catastrophically depart from within lane boundaries. A reference speed identical to the vehicle's initial speed of 50km/h resulted in perfect reference tracking. In contrast, for $v^\text{ref}=120$km/h, the vehicle departed from the road. Such a departure occured earliest for $v^\text{ref}= 94$km/h. Here a remark needs to be made. The inclusion of friction according to Section \ref{subsec_compTimeSpace} naturally lowers the reference speed (of originally 120km/h) and caused (if considered) the vehicle to remain within lane boundaries. In general though, the inclusion of \eqref{eq:LTV_vkmaxfric} still does not guarantee safe operation within road boundaries. This is because of the absence of \textit{corridor constraints}. While TOSS permits to easily enforce vehicle operation within road boundaries by \eqref{eq:OP_ey_constrts}, an equivalent formulation for a time-based method is more difficult. Because of the time parametrization, for the formulation of linearly constrained convex optimization problems, \textit{time-varying} polyhedral constraints have to be defined. An efficient logic for the design of such constraints that is not too conservative (too small polyhedra), and at the same time automated and sufficiently simple for applicability on arbitrarily curved road shapes is far from trivial. This holds generally for kinematic and dynamic vehicle models. In alternative to \eqref{eq:LTV-MPC}, an equivalent LTV-MPC problem can be formulated based on road-aligned but time-based dynamics \eqref{eq_def_dotsepsiey}. The simplest corresponding formulation of polyhedral constraints reads: $s_k^\text{min} \leq s_k \leq s_k^\text{max}$ and $e_{y,k}^\text{min} \leq e_{y,k} \leq e_{y,k}^\text{max},~\forall k=1,\dots,K$, with parameters $s_k^\text{min},s_k^\text{max},e_{y,k}^\text{min}$ and $e_{y,k}^\text{max}$ defining the admissible road segment at each time index $k$. Note, because of the time parametrization, the formulation is again strongly dependent on suitable and potentially time-varying reference velocities, that directly translate into the selection of aforementioned parameters. In contrast, for spatial-based methods, constraints \eqref{eq:OP_ey_constrts} hold \textit{time-invariantly}. For each $s_j$, the lateral admissible deviation is defined. Decisively, this is valid independently of \textit{when} and at \textit{what speed} each $s_j$ is reached. Finally, a remark to reference velocity generation. In ~\cite{urmson2008autonomous}, it is distinguished between four PWA designs: constant, linear, linear ramp and trapezoidal. A design logic is required to select one, \textit{and} to additionally determine suitable slope rates and velocity plateau levels that are dependent on the current vehicle state. Thus, for \textit{time-based} methods three tasks are required: reference path planning, reference velocity planning and reference tracking. The integration and synchronization of all of these methods is difficult, as emphasized in \cite{paden2016survey}. Instead of mission planning based on reference velocity assignments, the proposed spatial-based method encourages to work with \textit{waypoints}, and timings when these waypoints should be reached; suitable velocity trajectories are then returned automatically by \eqref{eq:LP}. This is regarded as a main benefit of TOSS. 

Examples 4-5 further illustrate characteristics. First, a road segment can be traversed significantly faster (Example 4). Second, while the LTV-method must operate the vehicle at its friction limits to achieve a corridor traversal time of $t^{\text{LTV},\star}=12.4$s, TOSS can achieve the same in a \textit{safer} fashion. Namely, the vehicle can be operated at velocities below its tire friction limits, see Fig. \ref{fig_EX5}. This is enabled by a) the minmax objective in \eqref{eq_OP_objFcn} resulting in trajectories exploiting the entire road width, and b) the spatiotemporal constraints \eqref{eq:OP_tWP_cstrts} for mission planning by means of waypoints.

To summarize, TOSS is built on the following key components: expressing time and vehicle dynamics in a road-aligned coordinate frame, an objective function combining final time $t_N$ and smooth steering (minmax objective), corridor constraints, a vehicle dimension constraint heuristic and friction constraints. Two consequences are characteristic. First, the combination of the steering-related part of the objective function \eqref{eq_OP_objFcn} in combination with corridor constraints \eqref{eq:OP_ey_constrts} causes a) smooth obstacle avoidance, and b) the exploitation of the complete admissible road width. A spatially-varying road width can be defined conveniently by $e_{y,j}^\text{min}$ and $e_{y,j}^\text{max}$ in \eqref{eq:LP}. Second, the incorporation of time dynamics \eqref{eq:OP_tjp1_eq} enables \eqref{eq:LP} to return a velocity profile. Instead of mission planning based on reference velocity assignments, the proposed spatial-based method thus encourages to work with \textit{waypoints}, and timings when these waypoints should be reached; suitable velocity trajectories are then returned automatically by \eqref{eq:LP}. Finally, we stress that the proposed method a) simultaneously generates \textit{both} velocity and steering trajectories, and b) is suitable for combining trajectory planning and tracking in \textit{one} step in a receding horizon control (RHC) scheme, applying velocity and steering commands directly to the vehicle's low-level controllers.

\section{Conclusion\label{sec_conclusion}}

We presented a linear programming-based tool for trajectory planning of automated vehicles by relating time to spatial-based system modeling, thereby enabling simultaneously time scheduling, time-optimal traversal of road segments and obstacle avoidance. We discussed the role of control rate constraints in a road-aligned coordinate frame. A heuristic constraint was presented to account for vehicle dimensions. We incorporated friction constraints in the second of two linear programs (LP) that are solved sequentially, whereby the solution of the first LP serves as input to the second LP. A comparison to a time-based method was given, illustrating the benefits of the proposed method. These include mission planning by means of waypoints and assignment of scheduling times at which these waypoints are meant to be traversed (spatiotemporal constraints), and the simultaneous generation of velocity and steering trajectories. 

The presented framework is expected to be also particularly useful for the control of automated vehicles at intersections~\cite{schildbach2016collision}, and for the coordination and scheduling of multi-vehicle systems, which is subject of ongoing work.

%

%
%
%
%
%
%

\nocite{*}
\bibliographystyle{ieeetr}
\bibliography{myref}



\end{document}